\documentclass[journal,final,onecolumn,12pt,twoside]{IEEEtranTCOM}
\usepackage{cite}
\usepackage{graphicx}
\usepackage{subfigure}
\usepackage{amsmath}
\usepackage{amssymb}
\usepackage{bm}
\usepackage{epsfig}
\usepackage{color}
\usepackage{url}
\usepackage{array}
\usepackage{multirow}
\usepackage{rotating}
\usepackage{extarrows}
\usepackage{algpseudocode}

\usepackage[ruled, vlined, linesnumbered]{algorithm2e}
\usepackage{tabularx}
\makeatletter
\newcommand{\removelatexerror}{\let\@latex@error\@gobble}
\makeatother

\usepackage{epstopdf}
\usepackage{float}
\def\IEEEQEDclosed{\mbox{\rule[0pt]{1.3ex}{1.3ex}}}
\def\@begintheorem#1#2{\par\bgroup{\scshape #1\ #2. }\it\ignorespaces}
\def\@opargbegintheorem#1#2#3{\par\bgroup%
   {\scshape #1\ #2\ ({\upshape #3}). }\it\ignorespaces}
\def\@endtheorem{\egroup}

\newtheorem{pavikl}{\textbf{Lemma}}

\makeatletter

\newcommand{\Rmnum}[1]{\expandafter\@slowromancap\romannumeral #1@}
\makeatother

\begin{document}
\title{Energy- and Spectral- Efficiency Tradeoff for D2D-Multicasts in Underlay Cellular Networks} 

\author{\IEEEauthorblockN{Ajay Bhardwaj and Samar Agnihotri}

\IEEEauthorblockA{School of Computing and Electrical Engineering, Indian Institute of Technology Mandi, HP - 175$\,$005, India}

Email: ajay\_singh@students.iitmandi.ac.in, samar@iitmandi.ac.in%
}
\maketitle
\begin{abstract}
Underlay in-band device-to-device (D2D) multicast communication, where the same content is disseminated via direct links in a group, 
has the potential to improve the spectral and energy efficiencies of cellular networks. However, most of the existing approaches for this 
problem only address either spectral efficiency (SE) or energy efficiency (EE). We study the trade-off between SE and EE in a single cell D2D integrated 
cellular network, where multiple D2D multicast groups (MGs) may share the uplink channel with multiple cellular users (CUs). We explore SE-EE trade-off for 
this problem by formulating the EE maximization problem with constraint on SE and maximum available transmission power. A power allocation algorithm is 
proposed to solve this problem and its efficacy is demonstrated via extensive numerical simulations. The trade-off between SE and EE as a function of 
density of D2D MGs, and maximum transmission power of a MG is characterized. 
\end{abstract}
\section{Introduction}
\label{Sec:1}
Supporting ever increasing number of mobile users with data-hungry applications, running on battery-limited devices, is posing a daunting challenge to 
telecommunications community. Underlay device-to-device (D2D) communication, which allows physically proximate mobile users to directly communicate with 
each other by reusing the spectrum and without going through the base station, holds promise to help us tackle this challenge \cite{Dfeng}. In a cellular 
network, underlay D2D communication offers opportunities for spectrum reuse and spatial diversity which may lead to enhanced coverage, higher throughput, 
and robust communication in the network \cite{liu2015device}. Further, for applications such as weather forecasting, live streaming, or file distribution, 
which may require the same chunk of data distributed to geographically proximate group of users, D2D multicasting may provide better utilization of network 
resources compared to D2D unicast or Base Station (BS) based multicast. However, extensive deployment of underlay D2D multicast in a network may cause 
severe co-channel interference due to spectrum reuse and rapid battery depletion of the multicasting D2D nodes due to higher transmit power to mitigate co-
channel interference and data-forwarding.

In cellular networks, two metrics namely spectrum efficiency (SE) and energy efficiency (EE) characterize how efficiently the spectrum and energy 
resources, respectively, are used. However, often conflicting nature of these metrics may not allow for simultaneous maximization of both in a network. 
There exists an extensive body of work that explores SE-EE trade-off in cellular networks\cite{tang2014resource,rao2016analytical} and some work that 
explore the nature of this trade-off for D2D communication \cite{gao2017energy,wei2016energy}. However, currently no systematic study of the nature of such 
trade-off for multiple D2D multicasts in underlay cellular networks exists. To the best of our knowledge, this letter provides the first such study.

Using stochastic geometry, we formulate a resource allocation problem that maximizes the EE of multiple D2D-multicasts in underlay cellular networks with 
constraints on SE and maximum transmission power. The formulated problem is non-convex, and is solved using the proposed heuristic gradient power 
allocation algorithm. We establish the trade-off between EE and SE with various network parameters such as density of D2D multicast transmitters, and 
maximum transmission power of MGs through numerical simulations.

\noindent\textit{Organization:} Section \ref{sec:2} introduces the system model. Section \ref{sec:3} and \ref{sec:4} introduce the problem formulation and the optimal power allocation algorithms, respectively. Performance evaluation is in Section \ref{sec:5}. Finally, Section \ref{sec:6} concludes the paper.
\section{System Model}
\label{sec:2}
A D2D-integrated cellular network is considered, where multiple D2D-multicast groups (MGs) may share the uplink channel with multiple cellular users (CUs).
Let $\mathcal{K} = \{1,2,\ldots,K\}$ denote the total number of orthogonal channels 
that can be shared by CUs and D2D MGs. Sharing of uplink channels is assumed instead of downlink because of 
asymmetric uplink and downlink traffic loads \cite{onireti2012energy}, and also as the eNB (evolved Node Base station) can handle interference 
effectively. The spatial distribution of CUs and MGs on the $k^\textrm{th}$ channel is modeled as  homogeneous 
Poisson Point Process $\Pi_{c,k}$ with density $\lambda_c^k$, and $\Pi_{g,k}$
with density $\lambda_g^k$, respectively, in the Euclidean plane $\mathbb{R}^2$. The proposed system model is an abstraction of a system where a single cell is divided into small cells, and multiple CUs may share a single channel. 

Let $ \vert \mathcal{U}_g \vert$ be the number of receivers in the $g^\textrm{th}$ MG ($ \vert \mathcal{U}_g \vert = 1$ corresponds to unicast communication.) We consider the 
variable number of receivers in every MG and assume that they always have data demands. The transmission powers of CU and D2D MG transmitter (D2D-Tx) on the $k^\textrm{th}$ channel are denoted by $p_{c,k}$ and $p_{g,k}$, respectively. In addition, total transmission power of CUs and D2D MGs is $P_C$ and $P_G$, i.e.
\\
\begin{subequations}
\setlength{\tabcolsep}{0pt}
\\
\noindent\begin{minipage}{0.35\hsize}
\begin{equation}\label{eq:1a}
    \sum\nolimits_{k=1}^{\vert \mathcal{K} \vert} p_{c,k} =  P_C,  \notag
    \addtocounter{equation}{1} 
\end{equation}
\end{minipage}(\ref{eq:1a}),
\begin{minipage}{0.35\hsize}
\begin{equation}\label{eq:1b}
  \sum\nolimits_{k=1}^{\vert \mathcal{K} \vert} p_{g,k} \leq  P_G   \notag
    \addtocounter{equation}{1} 
\end{equation}
\end{minipage}(\ref{eq:1b})
\\
\end{subequations}

For analysis, a reference receiver at the origin of cell (eNB for cellular and a typical D2D receiver for D2D communication) 
is considered. The radio propagation channel gain between the $i^\textrm{th}$ transmitter and the $j^\textrm{th}$ receiver, denoted as 
$h_{i,j}$, is assumed to be Rayleigh faded, and independently and identically exponentially distributed with unit mean, i.e. $h_{i,j}\sim \textrm{exp}(1)$. Therefore, the received power at the $j^\textrm{th}$ 
receiver is $p_j= p_i h_{i,j}d_{i,j}^{-\alpha}$, whereas $d_{i,j}$ denotes the distance between the $i^\textrm{th}$ and the $j^\textrm{th}$ node, and $\alpha$ is the path-loss exponent.

As we are considering the scenarios where a channel is shared by multiple CUs and multiple MGs, thus, the $r^\textrm{th} \left(r \in \mathcal{U}_g\right)$ D2D-Rx  experiences interference from co-channel CUs and other D2D MG-Tx. Therefore, the signal-to-interference-plus-noise-ratio (SINR) at the $r^\textrm{th} \left(r \in \mathcal{U}_g\right)$ D2D receiver on the $k^\textrm{th}$ channel is 
\begin{equation}
\label{eq:2}
\gamma_{g,r}^k =  \frac{p_{g,k} h_{g,r,k} d_{g,r}^{-\alpha}}
{\underset{c \in \Pi_{c,k}}{\sum} p_{c,k} h_{c,r,k} d_{c,r}^{-\alpha} + \underset{g' \in \Pi_{g,k}}{\sum}p_{g',k} h_{g',r,k} d_{g',r}^{-\alpha} +  N_0},
\end{equation}
As the system model is interference limited, therefore, thermal noise $N_0$ can be omitted, and hence \eqref{eq:2} becomes
\begin{equation}
\label{eq:3}
\gamma_{g,r}^k =  \frac{h_{g,r,k} d_{g,r}^{-\alpha}}{I_{c,r,k} + I_{g',g,k}},
\end{equation}
where $I_{c,r,k}  = \mathop{\sum}_{c \in \Pi_{c,k}} \frac{p_{c,k}}{p_{g,k}} h_{c,r,k} d_{j,r,k}^{-\alpha}$ and $I_{g',g,k} = \mathop{\sum}_{g' \in \Pi_{g,k}} \frac{p_{g,k}}{p_{g',k}}h_{g',r,k} d_{g',r,k}^{-\alpha}$.
In a MG, transmission rate is decided by channel conditions of the worst user \cite{meshgi2015joint}, so, SIR, and the corresponding date rate, respectively, are
 \begin{equation}
\label{eq:4}
\gamma_{g}^k =  \underset{r \in \mathcal{U}_g}{\min} \left(\frac{h_{g,r,k} d_{g,r}^{-\alpha}}{I_{c,r,k} + I_{g',g,k}}\right), ~~R_g^k = \log_2\left(1+ \gamma_{g}^k\right)
\end{equation}
An outage event for a MG g occurs if the aggregate rate of the $g^\textrm{th}$ group falls below its target rate, $R_g^{\textrm{th}}$. 
Therefore, the outage probability of the $g^\textrm{th}$ MG is given by the following lemma.
\begin{pavikl}
\label{Lemma:1}
The outage probability of a D2D MG communicating on the $k^\textrm{th}$ shared channel is 
\begin{equation}
\label{eq:5}
\textrm{Pr} \left(R_g^k < R_g^\textrm{th} \right) =  1 - \textrm{exp}{\Bigg \lbrace -\chi_{g,k} \left[ \lambda_c^k \left(\frac{p_{c,k}}{p_{g,k}}\right)^{\frac{2}{\alpha}}  + \lambda_g^k \right]\Bigg\rbrace},
\end{equation}
where $\chi_{g,k} = \pi  d_{g,r}^2 \Gamma \left( 1  + \frac{2}{\alpha}\right) \Gamma \left(1 - \frac{2}{\alpha}\right) \left(2^{\left(R_g^\textrm{th}/\vert \mathcal{U}_g \vert\right)} - 1 \right)^{{2}/{\alpha}}$,  $\Gamma(x) = \int_0^{\infty} t^{x-1} e^{-t} dt$ denotes complete gamma function.  
\end{pavikl}
\begin{IEEEproof}
Please refer to Appendix \ref{lemma_1_proof}
\end{IEEEproof}
Similarly, the SIR of a CU transmitting on the $k^\textrm{th}$ channel is
\begin{equation}
\label{eq:6}
\gamma_c^k = \frac{p_{c,k} h_{c,b} d_{c,b}^{-\alpha}}{\underset{c \in \Pi_{c,k}}{\sum} p_{c,k} h_{c',b,k} d_{c',b}^{-\alpha} + \underset{g \in \Pi_{g,k}}{\sum}p_{g,k} h_{g,b,k} d_{g,b}^{-\alpha}}
\end{equation}
with the corresponding outage probability given by the following lemma.
\begin{pavikl}
The outage probability of a CU on the $k^\textrm{th}$ channel, can be expressed as 
\begin{equation}
\label{eq:7}
\textrm{Pr}\left(R_c^k < R_c^\textrm{th}\right) = 1- \textrm{exp}\Bigg\lbrace -\chi_{c,k}\left[ \lambda_c^k + \lambda_g^k \left(\frac{p_{g,k}}{p_{c,k}}\right)^{\frac{2}{\alpha}}\right]\Bigg\rbrace
\end{equation}
where $R_c^\textrm{th}$ denotes the date rate threshold of CU and $\chi_c^k = \pi d_{c,b}^2 \Gamma\left(1+  \frac{2}{\alpha}\right)\Gamma\left(1-  \frac{2}{\alpha}\right) \left(2^{\left(R_c^\textrm{th}\right)} - 1 \right)^{{2}/{\alpha}}$.
\end{pavikl}
\begin{IEEEproof}
Similar to the proof of Lemma 1.
\end{IEEEproof}
These lemmas allow us to infer that
\begin{itemize}
\item The outage probability of D2D MGs increases with increase in MG geographical spread, $d_{g,r}$, because channel fading becomes too severe with increasing distance. 
\item The outage probability of D2D MGs decreases with 
decrease in densities of CUs and D2D MGs, this is due to mitigation of co-channel interference.
\item The outage of D2D MGs increases with increase in $p_{c,k}$, because it creates more interference to D2D 
transmission. While, higher $p_{c,k}$ decrease the outage of CUs. 
\end{itemize}
\section{Optimal Resource Allocation for D2D- Multicasts in Underlay Cellular Communication}
\label{sec:3}
The average throughput, $\text{SE}_g^k$, of D2D mulicast communication on the $k^\textrm{th}$ channel, can be expressed as \cite{kwon2015energy} :
\begin{equation}
\label{eq:8}
\text{SE}_g^k = \vert u_g \vert \lambda_g^k \textrm{log}_2 \left(1 +  \gamma_{g}^k\right) \textrm{Pr} \left(R_g^k \geq R_g^\textrm{th} \right)
\end{equation}

The EE is defined as the ratio of average throughput to the power consumption \cite{gao2017energy}. As in \cite{kwon2015energy}, 
the power consumption of D2D MGs which are communicating on the $k^\textrm{th}$ channel is $\lambda_g^k p_{g,k}$. 
Therefore, the total energy efficiency of D2D multicast underlay network is
\begin{equation}
\label{eq:9}
\text{EE}_g =  \sum\nolimits_{k=1}^{\vert \mathcal{K} \vert} \text{EE}_g^k =  \sum\nolimits_{k=1}^{\vert \mathcal{K} \vert}  \text{SE}_g^k/(\lambda_g^k p_{g,k})
\end{equation}
To ensure the high data rate to both CUs and D2D users, thresholds on outage probabilities of both the D2D users and CUs are put as follows:
\begin{equation}
\label{eq:10}
\textrm{Pr} \left(R_c^k < R_c^\textrm{th}\right) \leq \Theta_c, \quad  \text{and} \quad \textrm{Pr} \left(R_g^k < R_g^\textrm{th} \right) \leq \Theta_g, 
\end{equation}
where $\Theta_c$ and $\Theta_g$ denote the outage thresholds for cellular and D2D MGs transmissions, 
respectively. The transmit power in the $k^\textrm{th}$ channel 
should be less than the allowed upper bound for that channel, denoted as $p_{g,k}^\textrm{up}$. Thus, we have
\begin{equation}
\label{eq:11}
0 \leq p_{g,k} \leq p_{g,k}^\textrm{up}
\end{equation}From  \eqref{eq:5}, \eqref{eq:7}, and \eqref{eq:10}, feasible power region is 
\begin{align}
&p_{g,k} \geq p_{c,k} \left(\frac{- \textrm{ln}\left(1 - \Theta_{g}\right)}{\lambda_c^k\chi_{g,k}} - \frac{\lambda_g^k}{\lambda_c^k}\right)^{-\frac{\alpha}{2}},\\
& p_{g,k} \leq p_{c,k}\left(\frac{- \textrm{ln} \left(1 - \Theta_{c}\right)}{\lambda_g^k \chi_{c,k}} - \frac{\lambda_c^k}{\lambda_g^k}\right)^{\frac{\alpha}{2}} \label{p_sup}
\end{align}
%\end{equation}
Let $p_{g,k}^{\text{low}} = p_{c,k} \left(\frac{- \textrm{ln}\left(1 - \Theta_{g}\right)}{\lambda_c^k\chi_{g,k}} - \frac{\lambda_g^k}{\lambda_c^k}\right)^{-\frac{\alpha}{2}} \textrm{and \-} p_{g,k}^{\text{high}} = p_{c,k} \left(\frac{- \textrm{ln} \left(1 - \Theta_{c} \right)}{\lambda_g^k \chi_{c,k}} - \frac{\lambda_c^k}{\lambda_g^k}\right)^{\frac{\alpha}{2}}$. Therefore, infimum $p_{g,k}^{\inf}$ = $\max\{ 0$, $p_{g,k}^{\text{low}}\}$ and supremum $p_{g,k}^{\sup} = \min\lbrace p_{g,k}^\textrm{up},p_{g,k}^{\text{high}}\rbrace$ denote the upper and lower bound, respectively. With these tranformations, the EE optimization problem with constraints on SE and maximum transmission power is formulated as:
\begin{align}
\mathbf{P:}~~~ &\underset{p_{g,k}} {\max~} EE_g =  \sum\nolimits_{k=1}^{\vert \mathcal{K} \vert} EE_g^k \label{eq:14}\\
&\textrm{s.t. \-} p_{g,k}^{\inf} \leq p_{g,k} \leq p_{g,k}^{\sup}, \text{~and~}\sum\nolimits_{k=1}^{\vert \mathcal{K} \vert} p_{g,k} \leq  P_G \nonumber 
\end{align}
The problem \textbf{P} is non-convex as we prove in Appendix \ref{proof_of_convex}. From second constraint, one may infer that, as the CUs are primary users, so to maintain their priority over MG users, they are assumed to transmit at
full power, i.e. $P_c$. While D2D MG users are secondary users, and interference creators for cellular transmissions. Thus, they are assumed to transmit at lower powers, so that their sum power does not overshoot
the threshold $P_G$.

\section{Power Allocation For Optimal Resource Allocation}
\label{sec:4}
The objective function in optimization problem \textbf{P} is a summation of $\vert \mathcal{K} \vert$ 
functions. Let $p_{g,k}^*$ denote the global maximal point where single function $EE_g^k$ achieves its maximum, and if the sum power
constraint \eqref{eq:1b} is removed, then $p_{g,k}$ are mutually independent. The $EE_g$ achieves the maximum value when every $EE_g^k$
achieves its maximum. Intuitively, finding power that 
maximizes the $EE_g^k$, is easier than solving \textbf{P} as a whole. The power $p_{g,k}^*$ that maximizes the individual $EE_g^k$ can be found using the following lemma.
\begin{pavikl}
\label{lemma:4}
The value $p_{g,k}^*$ that maximizes the $EE_g^k$ is
\[
p_{g,k}^*= 
\begin{cases}
   p_{g,k}^{\sup},& p_{g,k}^{\sup} \leq \left(\frac{2Z_k^g}{\alpha}\right)^{\alpha/2}\\
    \left(\frac{2Z_k^g}{\alpha}\right)^{\alpha/2},  & p_{g,k}^{\inf} \leq \left(\frac{2Z_k^g}{\alpha}\right)^{\alpha/2} \leq p_{g,k}^{\sup} \\
    p_{g,k}^{\inf},  & p_{g,k}^{\inf} \geq \left(\frac{2Z_k^g}{\alpha}\right)^{\alpha/2} \\
\end{cases}
\]
where $Z_k^g=    \chi_g^k \lambda_c^k \left(p_{c,k}\right)^{\frac{2}{\alpha}}$
\end{pavikl}
\begin{proof}
Please refer to Appendix \ref{lemma_4_proof}
\end{proof}

Now, two cases arise, Case 1: when $\sum_{k=1}^{\vert \mathcal{K} \vert} p^{*}_{g,k} \leq P_G$

Then, the optimum power that is allocated is $p_{g,k}= p_{g,k}^*, \forall k = 1,\ldots,K$, and the maximum value of energy efficiency is $EE_g = \sum_{k=1}^{\vert \mathcal{K} \vert} EE_g^k(p_{g,k}^{*})$.

Case 2: when $\sum_{k=1}^{\vert \mathcal{K} \vert} p^{*}_{g,k} > P_G$.
Then, set $p_{g,k} = p_{g,k}^{*}, \forall k = 1,\ldots,K$, and update the value of $p_{g,k}$, such that power constraint \eqref{eq:1b} is maintained, while causing the least reduction in $EE_g$.
Let $p_{g,k}^{inst}$ be the instant value, having initial value of $p_{g,k}^{inst} = p_{g,k}^{*}$.
Let $\mathtt{d}$ denotes the difference between maximum available power and sum of assigned power. $\delta$ and $n$ denote the step size and number of steps, respectively, with $\delta = \mathtt{d} /n$. Parameter n controls the balance between computational effort and the performance, its value is assigned in accordance with convergence rate the error-tolerance requirements. As we are adjusting the power value which gives maximum value of the function, therefore, the function value decreases with reducing power, i.e, $EE_g^k\left(p_{g,k}^{inst} - \delta\right) < EE_g^k\left(p_{g,k}^{inst}\right)$.
To satisfy the equality constraint \eqref{eq:1b} while having the
least reduction in $EE_g$, we need to adjust that $p_{g,j},(j \in \mathcal{K}) $ for which $EE_g^j$ decreases the least after decreasing the instantaneous transmit power. 
\begin{align}
j &= \underset{\vert \mathcal{K} \vert} {\textrm{argmin}}~~~\vert EE_g^k\left(p_{g,k}^{inst}\right) - EE_g^k \left(p_{g,k}^{inst} - \delta\right) \vert \\
p_{g,j}^{inst}  &= p_{g,j}^{inst} - \delta
\end{align}
Iterating this process at least n times, leads to the sum power constraint \eqref{eq:1b} be met, and a near-to-optimal solution 
to \eqref{eq:14} is achieved. The formal description is given in Algorithm \rm{1}. The computation complexity of the proposed algorithm is $\mathcal{O}(n)$. 
\begin{algorithm}
\caption{The proposed power allocation algorithm}\label{EE_max_algo}
\DontPrintSemicolon
\KwIn{$K,n,P_G,P_C,\epsilon$, \textbf{Output:} $EE_g$}
\Begin($ $)
{
  $\textrm{Find} ~ p_{g,k}^* ~\textrm{from Lemma 3, and  assign}~ p_{g,k} =  p_{g,k}^*$\;
 $ \mathtt{d} =  P_G -  \sum_{k=1}^{K} p_{g,k} $ \;
 \uIf{$\mathtt{d} \leq 0$}{
 $EE_g = \sum_{k=1}^K EE_g^k\left(p_{g,k}^{*}\right)$
 }\Else{
  $\delta =  \mathtt{d} / n $\;
\While {$\vert P_G - \sum_{k=1}^K p_{g,k}\vert \geq \epsilon$}
{$j=\underset{\vert \mathcal{K} \vert} {\textrm{argmin}}\vert EE_g^k\left(p_{g,k}^{inst}\right) - EE_g^k\left(p_{g,k}^{inst} - \delta\right)\vert$ \;
\uIf{$ p_{g,j} + \delta > p_{g,j}^{\sup}\quad  \textrm{or} \quad p_{g,j} - \delta < p_{g,j}^{\inf}$}
{
$\vert\mathcal{K}\vert =\vert\mathcal{K}\vert\cap j$, \;
}
  \Else
  {
  $ p_{g,j} = p_{g,j} - \delta$, and
  $ \text{der}_j = \vert {EE_g^j}^{'} \left( p_{g,j}\right)\vert$
    }
 }} 
  \Return{$EE_g = \sum_{k=1}^K EE_g^k\left(p_{g,k}\right)$}
  }
\end{algorithm}
\section{Simulation Results}
\label{sec:5}
To explore the existence of trade-off between SE and EE for multiple D2D multicasts in underlay cellular networks and exploit it to design optimal resource allocation schemes, we carried out extensive numerical simulations. Some of the simulation parameters are as follows. The CU density $\lambda_c^k = \left[1e^{-4},1e^{-5},1e^{-4},1e^{-4},1e^{-4}\right]$ and D2D density $\lambda_g^k =\left[1e^{-3},1e^{-4},1e^{-3},1e^{-3},1e^{-3}\right]$ 
are considered, respectively. The CU density and D2D density are randomly chosen values. We used the following parameter values: $\alpha = 3, \vert u_g \vert =3$, $p_{g,k}^{\textrm{up}}= 15\textrm{dBm}$, $P_G = 25\textrm{dBm}$, $P_{c,k} = 26\textrm{dBm}$, $\Theta_c = 0.1$, and $\Theta_g=0.1$. The value of alpha has an impact on EE and SE, however, that is insignificant.

Fig. \ref{fig:2} depicts the behavior of EE with respect to required SE and the available D2D transmission power. It can be observed that 
for a given power level, with increase in SE requirement (1 to 10 bps/Hz), EE first increases then starts to decrease. This is because, 
for lower SE requirements, eNB tries to support many MGs per channel until outage thresholds are not violated. This leads to increase 
in sum rate and consequently, an increase in EE. While, for higher requirement of SE, eNB reduces the number of MG transmitters 
sharing a channel until outage probability constrains are not fulfilled. Therefore, sum rate decreases, consequently EE decreases.  
Similarly, increasing the transmission power of MG transmitter for fixed SE, higher data rate may be supported for small range D2D 
communication, therefore, the sum rate increases initially with power consumption. However, after some power threshold, MG 
transmitters start causing 
co-channel interference to CUs, and thus, CU outage probability starts increasing. To compensate this, eNB reduces the number of MGs 
to fulfil CU outage probability thresholds. Therefore, decrease in sum rate and consequently decrease in EE occurs.
 \begin{figure}
\centering
\includegraphics[width=3.5in]{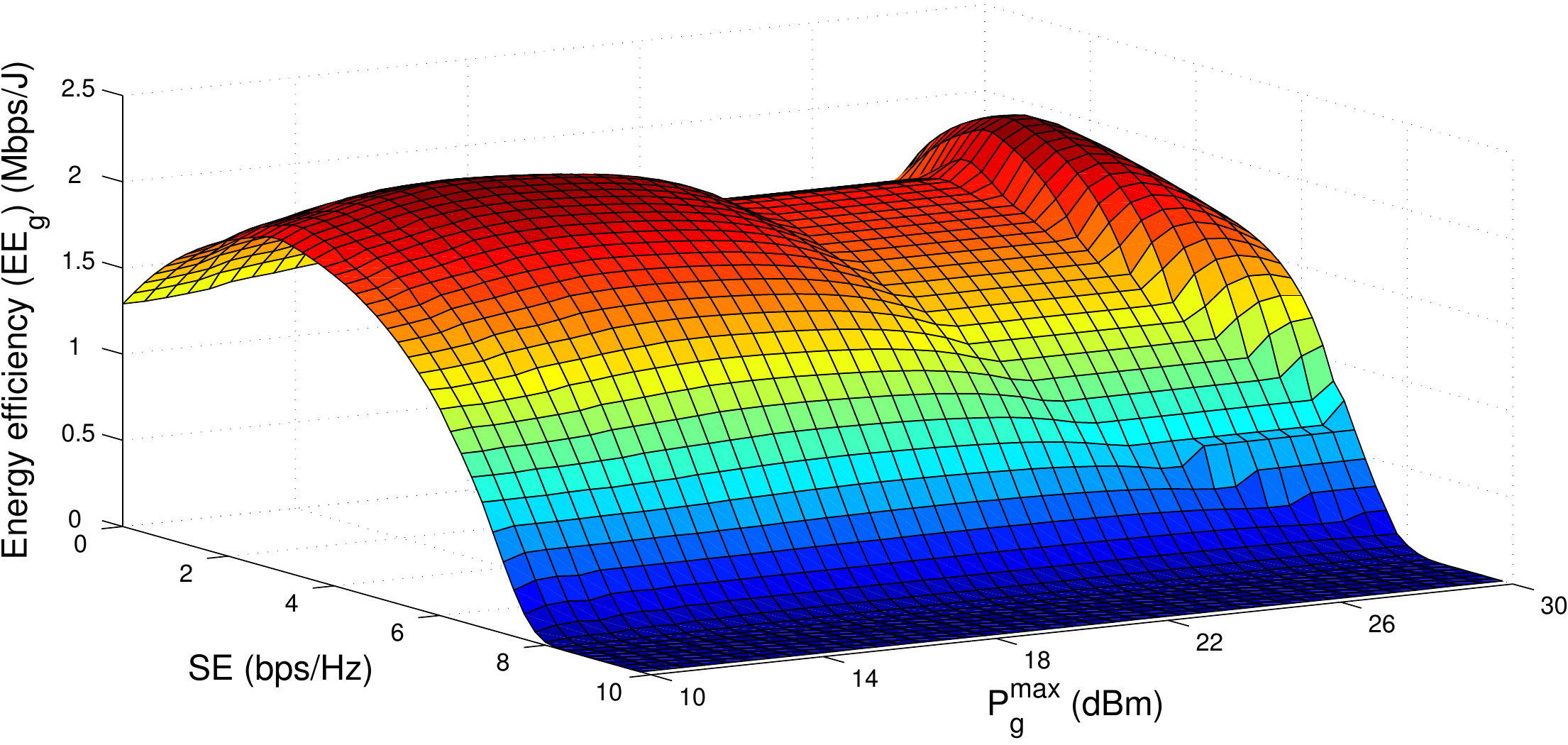}
\caption{EE as a function of SE and D2D users transmitter power}
\label{fig:2}
\end{figure} 
Fig. \ref{fig:3} depicts EE as a function of SE and D2D density. 
It can be observed that, for a given SE requirement, with increasing 
D2D density, EE first increases and then decreases. This is because, adding D2D users to the network (i.e increasing $\lambda_g$), 
results in an exponential increase in the average sum rate, and consequent increase in EE.  However, in high density, mutual 
interference starts increasing, 
and that limits the average sum rate per channel, leading to a decrease in EE. 
\begin{figure}
\centering
\includegraphics[width=3.5in]{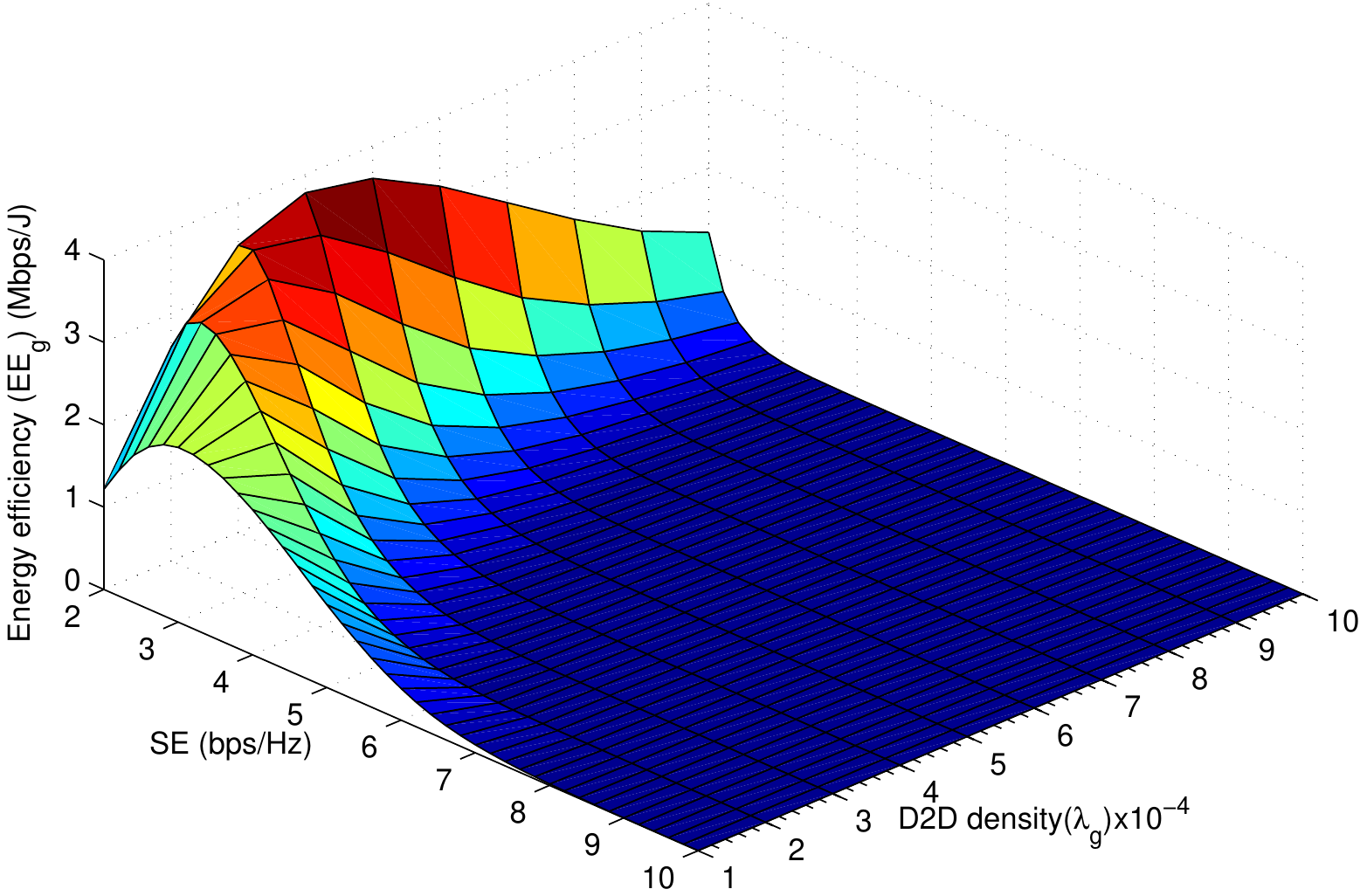}
\caption{EE as a function of SE and D2D users density}
\label{fig:3}
\end{figure}  
\section{Conclusion}
\label{sec:6}
For underlay D2D multicast in cellular networks, we addressed the energy efficiency and spectral efficiency trade-off. We assumed that multiple D2D-multicast group may share the channel with multiple CUs. Exact expression of average sum-rate and its relation with energy consumption is derived by utilizing the stochastic geometry. An energy optimization problem is formulated, having constraints on maximum power, and outage data rate. Our results showed that EE has different behavior with available power and spectral requirement. With increasing power, SE improves while EE initially increases then decreases. Similarly, with increase in D2D MG density, EE initially increases and then decreases. Indeed, for the EE, there is an optimal value of SE requirement that can be supported, which results in the maximal value of EE for each value of transmitter power of MG. 
\bibliographystyle{IEEEtran}

\begin{thebibliography}{1}
\providecommand{\url}[1]{#1}
\csname url@samestyle\endcsname
\providecommand{\newblock}{\relax}
\providecommand{\bibinfo}[2]{#2}
\providecommand{\BIBentrySTDinterwordspacing}{\spaceskip=0pt\relax}
\providecommand{\BIBentryALTinterwordstretchfactor}{4}
\providecommand{\BIBentryALTinterwordspacing}{\spaceskip=\fontdimen2\font plus
\BIBentryALTinterwordstretchfactor\fontdimen3\font minus
  \fontdimen4\font\relax}
\providecommand{\BIBforeignlanguage}[2]{{%
\expandafter\ifx\csname l@#1\endcsname\relax
\typeout{** WARNING: IEEEtran.bst: No hyphenation pattern has been}%
\typeout{** loaded for the language `#1'. Using the pattern for}%
\typeout{** the default language instead.}%
\else
\language=\csname l@#1\endcsname
\fi
#2}}
\providecommand{\BIBdecl}{\relax}
\BIBdecl

\bibitem{Dfeng}
D.~Feng, L.~Lu, Y.~Yuan-Wu, G.~Li, G.~Feng, and S.~Li, ``Device-to-device
  communications underlaying cellular networks,'' \emph{IEEE Trans. on
  Commun.}, vol.~61, no.~8, 2013.

\bibitem{liu2015device}
J.~Liu, N.~Kato, J.~Ma, and N.~Kadowaki, ``Device-to-device communication in
  lte-advanced networks: A survey,'' \emph{IEEE Commun. Surveys \& Tutorials},
  vol.~17, no.~4, 2015.

\bibitem{tang2014resource}
J.~Tang, D.~K. So, E.~Alsusa, and K.~A. Hamdi, ``Resource efficiency: A new
  paradigm on energy efficiency and spectral efficiency tradeoff,'' \emph{IEEE
  Trans. on Wireless Commun.}, vol.~13, no.~8, 2014.

\bibitem{rao2016analytical}
J.~B. Rao and A.~O. Fapojuwo, ``An analytical framework for evaluating
  spectrum/energy efficiency of heterogeneous cellular networks,'' \emph{IEEE
  Trans. on Vehicular Techn.}, vol.~65, no.~5, 2016.

\bibitem{gao2017energy}
H.~Gao, M.~Wang, and T.~Lv, ``Energy efficiency and spectrum efficiency
  tradeoff in the D2D-enabled hetnet,'' \emph{IEEE Trans. on Vehicular Techn.},
  vol.~66, no.~11, 2017.

\bibitem{wei2016energy}
L.~Wei, R.~Q. Hu, Y.~Qian, and G.~Wu, ``Energy efficiency and spectrum
  efficiency of multihop device-to-device communications underlaying cellular
  networks,'' \emph{IEEE Trans. on Vehicular Techn.}, vol.~65, no.~1, 2016.
  
\bibitem{onireti2012energy}
O.~Onireti, F.~H{\'e}liot, and M.~A. Imran, ``On the energy efficiency-spectral
  efficiency trade-off in the uplink of CoMP system,'' \emph{IEEE Trans. on
  Wireless Commun.}, vol.~11, no.~2, 2012.

\bibitem{meshgi2015joint}
H.~Meshgi, D.~Zhao, and R.~Zheng, ``Joint channel and power allocation in
  underlay multicast device-to-device communications,'' in \emph{Proc. IEEE
  ICC}, London, UK, June 2015.
  
\bibitem{sakr2015cognitive}
A.~H. Sakr and E.~Hossain, ``Cognitive and energy harvesting-based d2d
  communication in cellular networks: Stochastic geometry modeling and
  analysis,'' \emph{IEEE Tran. on Commun.}, vol.~63, no.~5, 2015.


\bibitem{kwon2015energy}
Y.~Kwon, T.~Hwang, and X.~Wang, ``Energy-efficient transmit power control for
  multi-tier MIMO hetnets,'' \emph{IEEE JSAC}, vol.~33, no.~10, 2015.
\end{thebibliography}

\appendices
\section{Proof of Lemma 1}
\label{lemma_1_proof}
\vspace{-0.1in}
\begin{align*}
&\textrm{Pr} \left(R_g^k < R_g^\textrm{th} \right) =  1 -  \textrm{Pr} \left( R_g^k \geq  R_g^\textrm{th} \right) \forall g \in \lambda_g^k \label{group_outage}\\
& = 1 - \text{Pr} \left(\gamma_g^k \geq 2^{ \frac{R_g^\textrm{th}} {\vert \mathcal{U}_g\vert}} -1 \right) =  1 - \text{Pr} \left(\frac{h_{g,r,k} d_{g,r}^{-\alpha}}{I_{c,r,k} + I_{g',g,k}}  \geq 2^ \frac{R_g^\textrm{th}}{\vert \mathcal{U}_g \vert} -1\right) \nonumber \\
& = 1-  \text{Pr} \left[ h_{g,r,k} \geq 2^ \frac{R_g^\textrm{th}}{\vert \mathcal{U}_g \vert - 1} d_{g,r}^{\alpha} \left(I_{c,r,k} + I_{g',g,k}\right) \right] \nonumber\\
 & =  1 - \Bigg \lbrace E \left(\mathop{\prod}_{j \in \Pi_{c,k}} \textrm{exp} \left( - \left(2^ \frac{R_g^\textrm{th}}{\vert \mathcal{U}_g \vert} - 1 \right) d_{g,r}^{\alpha} \frac{p_{c,k}}{p_{g,k}} h_{c,r,k} d_{c,g}^{-\alpha}\right)\right)  \nonumber \\
 & \times E \left( \mathop{\prod}_{g\in \Pi_{g,k}} \textrm{exp} \left(- \left(2^ \frac{R_g^\textrm{th}}{\vert \mathcal{U}_g \vert} - 1 \right)d_{g,r}^{\alpha} \frac{p_{g,k}}{p_{g',k}} h_{g',g,k} d_{g',g}^{-\alpha} \right)\right)\Bigg \rbrace\nonumber \\
 & = 1 -  \mathcal{L}_{I_{c,r,k} \left(h_{c,r,k}\right)}\left(T_{g,k}d_{g,r}^{\alpha}\right) \mathcal{L}_{I_{g',g,k} \left(h_{g',g,k}\right)} \left(T_{g,k} d_{g,r}^{\alpha}\right),
\end{align*}
where $T_{g,k} = \left(2^ \frac{R_g^\textrm{th}}{\vert \mathcal{U}_g \vert} - 1 \right)$.
According to definition of Laplace transform, we have 
\begin{align*}
&\mathcal{L}_{I_{c,r,k}\left(h_{c,r,k}\right)}\left(T_{g,k} d_{g,r}^{\alpha}\right) = \text{exp} \left[ - \lambda_g^k \left(\frac{p_{c,k}}{p_{g,k}}\right)^{\frac{2}{\alpha}} \pi T_{g,k}^{\frac{2}{\alpha}}d_{g,r}^2 \Gamma \left(1  + \frac{2}{\alpha}\right) \Gamma \left(1 - \frac{2}{\alpha}\right)\right]
\end{align*}
\begin{align*}
&\mathcal{L}_{I_{g',g,k}} \left(h_{g',g,k}\right)\left(T_{g,k} d_{g,r}^{\alpha}\right) = \text{exp} \left[ -\lambda_g^k \int_{0}^{\infty} E\left( h_{g',g,k} \right)\left(1 -  e^{-T_{g,k} d_{g,r}^{\alpha} r^{-\alpha}}\right) dr \right] \nonumber\\
&= \text{exp} \left[ -\lambda_g^k \left(\frac{p_{g,k}}{p_{g',k}}\right)^{\frac{2}{\alpha}} \pi T_{g,k}^{\frac{2}{\alpha}} d_{g,r}^2 \Gamma \left( 1 + \frac{2}{\alpha}\right) \Gamma \left(1 - \frac{2}{\alpha}\right)\right] 
\end{align*}

\section{Proof that the formulated problem $\mathbf{P}$ is non-convex}
\label{proof_of_convex}
Let $Y_k^g = \vert u_g \vert \textrm{log}_2 \left(1 + \gamma_g^k\right) \textrm{exp}\left( - \chi_g^k \lambda_g^k\right)$ and $Z_k^g =  \chi_g^k \lambda_c^k \left(p_{c,k}\right)^{\frac{2}{\alpha}}$, then we can write $EE_g^k$ as 
\begin{equation*}
EE_g^k \triangleq ({Y_k^g}/{p_{g,k}}) \exp\left( -  Z_k^g \left(\frac{1}{p_{g,k}}\right)^{\frac{2}{\alpha}}\right)
\end{equation*}

Taking double derivative of $EE_g^k$ with respect to $p_{g,k}$. 
\begin{align*}
\frac{d^2 \left(EE_g^k\right)}{d p_{g,k}^2} &=  2 Y_k^g \exp \left( - Z_k^g \left(\frac{1}{p_{g,k}}\right)\right) ^{2/\alpha}\left(\frac{1}{p_{g,k}} \right)^{3 + 4 /\alpha} \nonumber \\
& \times \left( p_{g,k}^{4 /\alpha} - \frac{Z_k^g}{\alpha} \left(\frac{2}{\alpha} + 3\right) p_{g,k}^{2 / \alpha}   + \frac{2 {\left(Z_k^g\right)}^2}{\alpha^2}\right) 
\end{align*}
In this equation, the 1st and 2nd terms are greater than zero, therefore, we only focus on the last term. 
Let $\nu  = p_{g,k}^{2/ \alpha}$, then the last term can be written as 
\begin{equation*}
\mathcal{V} =  \nu^2  - \frac{Z_k^g}{\alpha} \left(\frac{2}{\alpha} + 3\right) \nu + \frac{2 {\left(Z_k^g\right)}^2}{\alpha^2} 
\end{equation*}
The solutions $\nu_{1,k}$ and $\nu_{2,k}$ can be found as 
\begin{equation*}
\nu_{i,k} =  ({Z_k^g}/{2 \alpha^2}) \left(2 + 3\alpha \mp \sqrt{\alpha^2 + 12 \alpha + 4}\right), i \in \lbrace 1,2\rbrace
\end{equation*}
$\mathcal{V}$ is positive on the interval $\left(0,\nu_{1,k}\right) \cup \left(\nu_2, +\infty\right)$, and negative in interval $\left(\nu_{1,k}, \nu_{2,k}\right)$. Therefore,  $d^2 {\left(EE_g^k\right)}/{d p_{g,k}^2}$ is positive in interval $\left(0,\nu_{1,k}^{\alpha / 2}\right) \cup \left(\nu_2^{\alpha / 2}, +\infty\right),$
and negative in interval $\left(\nu_{1,k}^{\alpha / 2}, \nu_{2,k}^{\alpha / 2}\right)$. 
Therefore, $EE_g^k$ is convex on interval $\left(0, \nu_{1,k}^{\frac{\alpha}{2}}\right) \cup \left(\nu_{2,k}^{\frac{\alpha}{2}}, + \infty \right)$, but concave on the interval $\left(\nu_{1,k}^{\frac{\alpha}{2}},\nu_{2,k}^{\frac{\alpha}{2}}\right)$ 

\section{Proof of Lemma 3}
\label{lemma_4_proof}
Take the derivative of $EE_g^k$ with respect to $p_{g,k}$, 
\begin{equation*}
\frac{d \left(EE_g^k\right)}{d p_{g,k}} =  \frac{Y_k^g}{p_{g,k}^2} \exp \left(  - Z_k^g \left(\frac{1}{p_{g,k}^{2/\alpha}}\right)\right) \times \left( \frac{2 Z_k^g}{\alpha}\left(\frac{1}{p_{g,k}^{2/\alpha}}\right) - 1 \right)
\end{equation*}
From this equation, it can be observed that derivative of $EE_g^k$ is positive if $p_{g,k}$ lies in interval $\left(0, \left(\frac{2 Z_k^g}{\alpha}\right)^{\alpha/2} \right)$, and negative if $p_{g,k}$ lies in interval $\left(\left(\frac{2 Z_k^g}{\alpha}\right)^{\alpha/2}, + \infty\right)$, and reaches the global maxima at $p_{g,k} = \left(\frac{2 Z_k^g}{\alpha}\right)^{\alpha/2}$. Therefore, three feasible region exist, $EE_g^k$ is increasing monotonically in  $\left(0, \left(\frac{2 Z_k^g}{\alpha}\right)^{\alpha/2} \right)$,  decreasing monotonically in $\left(\left(\frac{2 Z_k^g}{\alpha}\right)^{\alpha/2}, + \infty\right)$, and having maximal point at  $p_{g,k} = \left(\frac{2 Z_k^g}{\alpha}\right)^{\alpha/2}$. If $p_{g,k}^{\sup} \leq\left(\frac{2 Z_k^g}{\alpha}\right)^{\alpha/2}$, then $EE_g^k$ is an increasing function, and, it reaches the maximum value when $p_{g,k}^{*} =  p_{g,k}^{\sup}$. Second region is $ p_{g,k}^{\inf} \leq \left(\frac{2 Z_k^g}{\alpha}\right)^{\alpha/2} \leq p_{g,k}^{\sup}$, then the maximal point of $EE_g^k$ is within feasible region. Therefore, optimal point is $p_{g,k}^* =   \left(\frac{2 Z_k^g}{\alpha}\right)^{\alpha/2}$. Third region is, if $ p_{g,k}^{\inf} \geq \left(\frac{2 Z_k^g}{\alpha}\right)^{\alpha/2}$, then $EE_g^k$ decreases monotonically in the feasible region. Hence, the maximum value of $EE_{g,k}$ is achieved at $p_{g,k}^*  = p_{g,k}^{\inf}$   
\end{document}